\documentclass[11pt]{article}
\usepackage{textcomp}
\usepackage{xcolor} % Package to handle colors

\usepackage{cite} % Cites and compresses citation numbers

\usepackage[pdftex]{graphicx}
\graphicspath{{../figures/}}
\DeclareGraphicsExtensions{.pdf} %,.jpeg,.png}

\usepackage{amsmath}
\usepackage{amsfonts}
\usepackage{amssymb}

\usepackage{algorithmic}
\usepackage{array} % Enhances array, tabular environments (strongly encouraged)

\usepackage{stfloats} % Double column floats placed at the bottom
\usepackage{url} % Handling and breaking URLs

\makeatletter
    \@ifpackageloaded{xcolor}{}{\usepackage{xcolor}}
\makeatother
\makeatletter
    \@ifpackageloaded{needspace}{}{\usepackage{needspace}}
\makeatother

\newcommand{\uppercaseGreek}[1]{%
    \begingroup
    \ucmathlist\MakeUppercase{#1}%
    \endgroup
}

\newcommand{\ucmathlist}{%
    \def\alpha{\mathrm{A}}%
    \def\beta{\mathrm{B}}%
    \let\gamma=\Gamma
    \let\delta=\Delta
    \def\epsilon{\mathrm{E}}%
    \def\varepsilon{\mathrm{E}}%
    \def\zeta{\mathrm{Z}}%
    \def\eta{\mathrm{H}}%
    \let\theta=\Theta
    \let\vartheta=\Theta
    \def\iota{\mathrm{I}}%
    \def\kappa{\mathrm{K}}%
    \let\lambda=\Lambda
    \def\mu{\mathrm{M}}%
    \def\nu{\mathrm{N}}%
    \let\xi=\Xi
    \let\pi=\Pi
    \let\varpi=\Pi
    \def\rho{\mathrm{P}}%
    \def\varrho{\mathrm{P}}%
    \let\sigma=\Sigma
    \def\tau{\mathrm{T}}%
    \let\upsilon=\Upsilon
    \let\phi=\Phi
    \let\varphi=\Phi
    \def\chi{\mathrm{X}}%
    \let\psi=\Psi
    \let\omega=\Omega
}

\makeatletter
    \@ifpackageloaded{amsthm}{}{

        \usepackage{amsthm}
    }
\makeatother
\theoremstyle{plain}
    \newtheorem{theorem}{Theorem}
    \newtheorem{proposition}[theorem]{Proposition}
    
    \newtheorem{corollary}[theorem]{Corollary}
\theoremstyle{definition}
    \newtheorem{definition}{Definition}
    \newtheorem{remark}{Remark}
    \newtheorem{assumption}{Assumption}

\makeatletter
\def\renewtheorem#1{%
    \expandafter\let\csname#1\endcsname\relax
    \expandafter\let\csname c@#1\endcsname\relax
    \gdef\renewtheorem@envname{#1}
    \renewtheorem@secpar
}
\def\renewtheorem@secpar{\@ifnextchar[{\renewtheorem@numberedlike}{\renewtheorem@nonumberedlike}}
\def\renewtheorem@numberedlike[#1]#2{\newtheorem{\renewtheorem@envname}[#1]{#2}}
\def\renewtheorem@nonumberedlike#1{  
    \def\renewtheorem@caption{#1}
    \edef\renewtheorem@nowithin{\noexpand\newtheorem{\renewtheorem@envname}{\renewtheorem@caption}}
    \renewtheorem@thirdpar
}
\def\renewtheorem@thirdpar{\@ifnextchar[{\renewtheorem@within}{\renewtheorem@nowithin}}
\def\renewtheorem@within[#1]{\renewtheorem@nowithin[#1]}
\makeatother
\input{auxFiles/symbolDef.sty}

\usepackage{amsmath}
\usepackage{amssymb}
\usepackage{amsthm}
\usepackage{todonotes}
\usepackage{tikz}
\usepackage{pgfplots}
\usepackage{graphicx}
\usepackage{epstopdf}
\usepackage[labelfont=bf,font=footnotesize]{caption}
\usepackage[labelfont=bf,font=footnotesize]{subcaption}
\usepackage{xcolor}
\usepackage{bbm}
\usepackage{makecell}
\usepackage{authblk}
\usepackage{fullpage,etoolbox}
\usepackage[colorlinks=true,citecolor=blue,linkcolor=blue,urlcolor=blue]{hyperref}
\usepackage[shortlabels]{enumitem}
\usepackage[round,sort,compress]{natbib}
\usepackage{listings}
\usepackage{balance}  % Balance the columns

\begin{document}

\title{Distributed Optimal Control of Graph Symmetric Systems via Graph Filters}

%%%%% ALL
\author{Fengjun Yang$^{\ast}$, Fernando Gama$^{\dag}$, Somayeh Sojoudi$^{\ddag}$ and Nikolai Matni$^{\S}$
\thanks{$^\ast$F. Yang is with the Dept. of Comput. and Info. Sci, University of Pennsylvania, Philadelphia, PA. $^\dag$F. Gama is with the Dept. of Elect. and Comput. Eng., Rice University, Houston, TX. $^\ddag$S. Sojoudi is with the Dept. of Elect. Eng. and Comput. Sci, University of California, Berkeley, CA. $^\S$N. Matni is with the Dept. of Elect. and Syst. Eng., Philadelphia, PA.}
\thanks{N.M. is supported by NSF awards CPS-2038873, CAREER award ECCS-2045834, and a Google Research Scholar award. F.Y. is in part supported by NSF CAREER award ECCS-2045834. F.G. and S.S. are supported by ONR and NSF grants.}%
% sojoudi@berkeley.edu, nmatni@seas.upenn.edu, fengjun@seas.upenn.edu, fgama@rice.edu
}
%%%%%%%%%%%%%%%%%%%%%%%%%%%%%%%%%%%%%%%%
%%%%           Make Title           %%%%
%%%%%%%%%%%%%%%%%%%%%%%%%%%%%%%%%%%%%%%%

\maketitle

%%%%%%%%%%%%%%%%%%%%%%%%%%%%%%%%%%%%%%%%
%%%%            ABSTRACT            %%%%
%%%%%%%%%%%%%%%%%%%%%%%%%%%%%%%%%%%%%%%%

\begin{abstract}
    Designing distributed optimal controllers subject to communication constraints is a difficult problem unless structural assumptions are imposed on the underlying dynamics and information exchange structure, e.g., sparsity, delay, or spatial invariance. In this paper, we borrow ideas from graph signal processing and define and analyze a class of \emph{Graph Symmetric Systems} (GSSs), which are systems that are symmetric with respect to an underlying graph topology. We show that for linear quadratic problems subject to dynamics defined by a GSS, the optimal centralized controller is given by a novel class of graph filters with transfer function valued filter taps and can be implemented via distributed message passing. We then propose several methods for approximating the optimal centralized graph filter by a distributed controller only requiring communication with a small subset of neighboring subsystems.  We further provide stability and suboptimality guarantees for the resulting distributed controllers. Finally, we empirically demonstrate that our approach allows for a principled tradeoff between communication cost and performance while guaranteeing stability.  Our results can be viewed as a first step towards bridging the fields of distributed optimal control and graph signal processing.
\end{abstract}

%%%%%%%%%%%%%%%%%%%%%%%%%%%%%%%%%%%%%%%%
%%%%            KEYWORDS            %%%%
%%%%%%%%%%%%%%%%%%%%%%%%%%%%%%%%%%%%%%%%

%\begin{IEEEkeywords}
%    distributed control, linear-quadratic control, graph neural networks
%\end{IEEEkeywords}

\section{Introduction} \label{sec:intro}

Computing a distributed optimal controller in which subcontrollers have access to subsets of global system information is in general a computationally intractable problem.  Indeed, even when restricted to quadratic costs, Gaussian noise, and linear dynamics, the resulting optimal controller can be nonlinear and difficult to compute \citep{witsenhausen1968counterexample}.  Nevertheless, significant progress has been made in distributed optimal controller synthesis over the past two decades by identifying structural assumptions on the underlying dynamics and information exchange structure such that the resulting distributed controller synthesis problem is convex. 

One such structural assumption that has been shown to lead to tractable distributed optimal control problems is spatial invariance \citep{bamieh2002distributed} (and other closely related notions of symmetry \citep{massioni2009distributed}). Such systems are invariant under subsystem permutations, and have been shown to have optimal centralized controllers that are approximately distributed.  In particular, this allows for distributed controllers that enjoy stability and near-optimality guarantees to be computed by appropriately truncating the centralized controller.

\emph{Contributions:} In this paper, inspired by results from graph signal processing, we introduce the notion of Graph Symmetric Systems (GSSs), which are systems that are symmetric with respect to an underlying graph topology (formalized in \S\ref{sec:LQR}-A).  We show that for such systems, the resulting Linear Quadratic (LQ) centralized optimal controller admits an efficient message passing implementation in the form of a novel class of graph filters defined by transfer function filter taps.  We subsequently propose and analyze two complementary approaches to computing near-optimal distributed controllers by truncating the centralized optimal controller subject to stability constraints.  By leveraging tools from robust System Level Synthesis (SLS) \citep{matni2017scalable, anderson_system_2019}, we show that these truncation algorithms can be solved via convex optimization, and that the resulting distributed controllers enjoy sub-optimality guarantees relative to the centralized optimal controller.  These results constitute an important first step towards bridging the complementary, but traditionally disparate, fields of distributed optimal control and graph signal processing.

\emph{Related work:} An alternative structural assumption for tractable distributed optimal control of linear systems can be specified in terms of the sparsity and delay patterns of the control system.  In particular, it is possible to characterize conditions on the sparsity and delay patterns of the information exchanged between subcontrollers relative to the propagation of signals through sparse and delayed distributed plants such that distributed optimal control is tractable.  The seminal paper \cite{rotkowitz2005characterization} introduced the notion of quadratic invariance,\footnote{We note that spatially invariant systems, as defined in~\cite{bamieh2002distributed}, also satisfy quadratic invariance.  We show in Appendix B that graph symmetric systems and controllers lead to optimal control problems satisfying quadratic invariance.} which built upon and generalized funnel causality \citep{bamieh2005convex}, showed that so long as subcontrollers could communicate as quickly as control signals propagated through the plant, then the resulting distributed optimal control problem could be solved via convex optimization.  This convex parameterization of sparse and delayed controllers has since been further generalized in the System Level Synthesis (SLS) \citep{anderson_system_2019} and Input-Output Parameterization (IOP) \citep{furieri2019input} frameworks, which allow for even richer classes of sparsity and delay patterns to be imposed on distributed controllers.

A related class of distributed controllers are those based on Graph Neural Networks (GNNs).  GNNs can be viewed as graph filters followed by pointwise nonlinear activation functions \citep{Ruiz2021-GNN}, and among other favorable properties, enjoy stability to graph perturbations \citep{Gama2020-Stability}.  While recent use of GNNs for distributed control has shown promise \citep{Gama2022-DistributedLQR, Gama2022-ControlGNN, yang2021communication}, such results currently lack strong guarantees of stability.  We believe the results in this paper are a first step towards addressing this gap in the literature, by explicitly connecting graph filters and distributed optimal controllers. The direct relationship between graph filters and GNNs suggests that understanding the former will give insight in the effects of the latter.

%Informally, a GSS specified by a graph representation matrix $S$ (e.g., adjacency matrix or graph Laplacian) is a linear system for which all matrices specifying the dynamics and cost matrices are simultaneously diagonalized by the eigenvalue matrix $S$ -- we refer to matrices satisfying this property as \emph{graph symmetric} with respect to S. By leveraging the SLS framework, we show that the optimal centralized controller of a GSS is itself graph symmetric with respect to $S$.
\textit{Notation}: We use upper- and lower-case letters such as $A$ and $x$ to denote matrices and vectors respectively, although lower-case letters might also be used for scalars or functions (the distinction will be apparent from the context). For both upper- and lower-case letters, we use \textit{boldface} such as $\mtLambda$ and $\boldsymbol\phi$  to denote transfer matrices or vector/scalar transfer functions.
\section{The Linear Quadratic Regulator Problem for Graph Symmetric Systems} \label{sec:LQR}
Consider a discrete-time linear time-invariant (LTI) system composed of $N$ interconnected scalar subsystems, each with state $x_i(t) \in \fdR$, control input $u_i(t) \in \fdR$ and which evolves under the dynamics
\begin{equation}\label{eq:subsys-dynamics}
    x_i(t+1) = \sum_{j=1}^{N}A_{ij}x_j(t) + \sum_{j=1}^{N}B_{ij}u_j(t) + w_i(t),
\end{equation}
for suitable matrices $A_{ij}, B_{ij}$ describing the interaction between subsystems. Here $w_i(t)$ is an i.i.d. zero-mean noise. We can compactly express the dynamics of the full system in terms of the joint states $x(t) = [x_1(t), \ldots, x_N(t)]^\top$ and joint control actions $u(t) = [u_1(t), \ldots,u_N(t)]^\top$ as
\begin{equation} \label{eq:linearDynamics}
    x(t+1) = A x(t) + B u(t) + w(t),
\end{equation}
where $(A,B)$ are defined such that the global dynamics \eqref{eq:linearDynamics} are consistent with the subsystem dynamics \eqref{eq:subsys-dynamics}.

Our goal is to find a (potentially time-varying) state-feedback controller $K_t$ that minimizes the cost
% eq:quadraticCost
\begin{equation} \label{eq:quadraticCost}
    J \big( \{x(t)\}, \{u(t)\} \big) := \lim_{T \to \infty} \frac{1}{T}\sum_{t=0}^{T-1} \xp_{w} \big[ x(t)^{\Tr} Q x(t) + u(t)^{\Tr} R u(t) \big]
\end{equation}
where $u(t) = K_t(x(t))$ and $Q \succeq 0, R \succ 0$ are known symmetric $N \times N$ matrices. The Linear Quadratic Regulator (LQR) problem is then given by
\begin{equation} \label{eq:LQG}
    \min_{\mtK} J\big( \{x(t)\},\{u(t)\} \big) \quad \text{s. t. } u(t) = K_t\big(x(t)\big).
\end{equation}
In the centralized setting where each subsystem has access to the global state, it is well-known that the controller that solves \eqref{eq:LQG} is a linear static controller $u(t) = K^{\opt} x(t)$ where $K^{\opt} = -(R+ B^{\top} P B)^{-1}B^{\top} P A$ and $P$ is the unique solution to the discrete-time algebraic Ricatti equation:
% eq:Ricatti
\begin{equation} \label{eq:Ricatti}
    P = A^{\top} P A - A^{\top} P B \big( R+ B^{\top} P B \big)^{-1} B^{\top} P A + Q.
\end{equation}
In this work, we consider a distributed variant of \eqref{eq:LQG} where each subsystem can only exchange information with a small subset of subsystems. Specifically, this communication constraint is encoded as a graph $\graph = \{\stV, \stE\}$, where $\stV = \{\lmv_{1},\ldots,\lmv_{N}\}$ is the set of $N$ components (nodes) and where $\stE \subseteq \stV \times \stV$ is the set of the corresponding interconnections (edges). It is assumed that the graph is undirected, i.e. $(\lmv_{i},\lmv_{j}) \in \stE$ if and only if $(\lmv_{j},\lmv_{i}) \in \stE$. As described in the introduction, general information exchange constraints can lead to non-convex optimal control problems \citep{witsenhausen1968counterexample}. However, as we show later, under suitable graph symmetry assumptions on the dynamics matrices $A, B$ and the cost matrices $Q, R$, the optimal centralized controller admits a distributed message passing implementation allowing for a principled tradeoff between communication complexity and controller performance.

In the rest of this section, we borrow ideas from graph signal processing \citep{Sandryhaila2013-DSPG} and introduce the notion of \textit{graph symmetric systems}. First, we introduce a convenient way to define operations that respect the underlying communication graph structure via the graph matrix description (GMD) $S \in \fdR^{N \times N}$.
The matrix $S$ is such that the $(i,j)^{\text{th}}$ entry is zero whenever there is no connection between components $\lmv_{i}$ and $\lmv_{j}$, i.e. $[S]_{ij}=0$ if $i \neq j $ and $(\lmv_{i},\lmv_{j}) \notin \stE$.
Note that, since the graph is undirected, the matrix $S$ is symmetric.
Therefore, it has an eigedecomposition in terms of an orthonormal basis of eigenvectors $S = V \Lambda_{S} V^{\top}$ where $\Lambda_{S} \in \fdR^{N \times N}$ is a diagonal matrix with elements $\lambda_{S,i} \in \fdR$ such that $S v_{i} = \lambda_{S,i} v_{i}$ for $v_{i} \in \fdR^{N}$ being the $i^{\text{th}}$ column of $V$. We now introduce the notion of a graph symmetric system.

\begin{definition}[Graph Symmetric System]\label{def:distributedLinear}
Given a GMD $S=V \Lambda_{S} V^{\top}$ for a graph $\mathcal G$, a linear system \eqref{eq:linearDynamics} is \textit{graph symmetric} with respect to $\mathcal G$ if the dynamics matrices $A, B$ are simultaneously diagonalized by $V$, i.e.,
    % eq:distributedLinear
    \begin{equation} \label{eq:distributedLinear}
    \begin{aligned}
        A = V \Lambda_{A} V^{\top} \quad &, \quad B = V \Lambda_{B} V^{\top},
    \end{aligned}
    \end{equation}
where $\Lambda_{A}, \Lambda_{B}$ are diagonal.
\end{definition}

Note that Definition \ref{def:distributedLinear} does not require the dynamics to be sparse. In fact, matrices $A$ and $B$ of the form in Definition \ref{def:distributedLinear} can be arbitrarily dense, i.e., the evolution of a subsystem state $x_i(t+1)$ can depend on subsystems that $i$ cannot directly communicate with \citep{Gama2019-LinearControl}. This is distinct from sparsity/delay constraints used in \cite{rotkowitz2005characterization, anderson_system_2019, furieri2019input}, and encodes a different notion of symmetry than that exploited in the distributed control of spatially invariant systems \citep{bamieh2002distributed}.

By well-known results in graph signal processing \citep{Sandryhaila2013-DSPG}, simultaneous diagonalizability of the system matrices $(A,B)$ and the GMD $S$ implies\footnote{Under the assumption that $S$ corresponds to a finite graph and has all distinct eigenvalues. On a high level, the result follows directly from the Cayley-Hamilton theorem.} that they can be written as matrix polynomials of $S$ of degrees at most $N-1$,
% eq:polynomialOfS
\begin{equation} \label{eq:polynomialOfS}
    A = \sum_{k=0}^{N-1}h_{A,k}S^k \quad , \quad B = \sum_{k=0}^{N-1}h_{B,k}S^k.
\end{equation}
Matrices that can be expressed in this matrix polynomial forms are called \textit{graph filters} \citep{Segarra2017-GraphFilterDesign} and the coefficients $h_{A,k}, h_{B,k}$ are referred to as the filter weights or \textit{filter taps}.
%Simultaneous diagonalizability of the system matrices $(A,B)$ and the GMD $S$ allows them to be written as \textit{graph filters} \cite{Segarra2017-GraphFilterDesign}. By well-known results in graph signal processing \cite{Sandryhaila2013-DSPG}, Definition \ref{def:distributedLinear} implies\footnote{Under the assumption that $S$ corresponds to a finite graph and has all distinct eigenvalues. On a high level, the result follows directly from the Cayley-Hamilton theorem.} that $A$ and $B$ can be written as matrix polynomials of $S$ of degrees at most $N-1$,
% eq:polynomialOfS
%\begin{equation} \label{eq:polynomialOfS}
%    A = \sum_{k=0}^{N-1}h_{A,k}S^k \quad , \quad B = \sum_{k=0}^{N-1}h_{B,k}S^k.
%\end{equation}
%
%Matrices that can be expressed in this matrix polynomial forms are called \textit{graph filters} and the coefficients $h_{A,k}, h_{B,k}$ are referred to as the filter weights or \textit{filter taps}.
%Note that, when $A$ and $B$ can be written as a polynomial of $S$, this implies that \red{the dynamics evolving from $x(t)$ to $x(t+1)$ rely entirely on linear combinations carried out between components that share a connection, and hence, the system dynamics can be viewed as implementing distributed message passing \cite{Ruiz2021-GNN}.}
%To see this, consider first the operation $S x(t)$ whose $i^{\text{th}}$ entry yields

We now give a message-passing interpretation of graph symmetric systems. First, it can be seen from the sparsity pattern of the GMD $S$ that the output of $Sx(t)$ can be computed entirely as a linear combination of the states in nodes $1$ hop away in $\mathcal G$. To see this, consider the operation $S x(t)$ whose $i^{\text{th}}$ entry yields
% eq:graphShift
\begin{equation} \label{eq:graphShift}
    [S x(t)]_{i} = \sum_{j:(\lmv_{j},\lmv_{i}) \in \stE} [S]_{ij} [x(t)]_{j}.
\end{equation}
%
%It can be seen from the sparsity pattern of the GMD $S$ that the output of $Sx(t)$ can be computed entirely as a linear combination of the states in nodes $1$ hop away in $\mathcal G$.
More generally, when considering polynomials, it is observed that $S^{k}x(t)$ is equivalent to exchanging $k$ times information with one-hop neighbors.
Therefore, if the system matrix $A$ and the control matrix $B$ are polynomials of $S$, then the evolution of the system can be computed entirely by means of exchanges with neighboring nodes. Hence, the system dynamics can be viewed as implementing distributed message passing \citep{Ruiz2021-GNN}.
Examples of such linear, distributed systems, include both discrete-time and continuous-time diffusions, solutions to the heat equation, among many others, see \cite{Gama2019-LinearControl} and references therein.

For the rest of the paper, we also assume that the cost matrices for the LQR problem \eqref{eq:LQG} can also be simultaneously diagonalized with the dynamics matrices. Formally, we make the following assumption.
\begin{assumption}\label{assump:diagonalizableCost}
The system \eqref{eq:linearDynamics} defines a graph symmetric system with respect to a fixed GMD $S$, and the cost matrices $(Q,R)$ defining the LQR problem \eqref{eq:LQG} are graph symmetric with respect to $S$, i.e., they are simultaneously diagonalized by the ortho-bases $V$ satisfying $S=V\Lambda_S V^\top$.  In particular
\[
Q = V \Lambda_{Q} V^{\top} \quad, \quad R = V \Lambda_{R} V^{\top},
\]
where $\Lambda_Q, \Lambda_R$ are symmetric.
\end{assumption}
\section{Optimal Distributed Linear Controller via System Level Synthesis} \label{sec:SLS}

SLS provides a convex parameterization of achievable closed-loop system responses \citep{wang_separable_2018,anderson_system_2019}, which can be leveraged to show that the optimal controller for graph symmetric systems under Assumption \ref{assump:diagonalizableCost} can be written as a novel class of graph filters defined by transfer function valued filter taps.
% In this section, we provide a short summary of SLS \cite{wang_separable_2018,anderson_system_2019}, which  We then leverage SLS to show that the optimal controller for graph symmetric systems under Assumption \ref{assump:diagonalizableCost} can be written as a novel class of graph filters defined by transfer function valued filter taps.

\subsection{Background: System-Level Synthesis}

As noted in \S4 of \cite{anderson_system_2019}, we can compactly write the system dynamics \eqref{eq:linearDynamics} in the frequency domain as
\[(zI - A)\mathbf x = B \mathbf u + \mathbf w,\]
where $\mathbf x = \sum_{t=0}^{\infty} z^{-t}x(t)$ is the signal $x(t)$ in the $z$-domain, and idem for $\mathbf u$ and $\mathbf w$. For a (dynamic) linear state-feedback controller $\mathbf{u} = \mathbf K \mathbf x$, it follows immediately that
\begin{equation}\label{eq:system-response}
    \begin{aligned}
    \vcx &= (zI - A - B\mtK)^{-1} \vcw =: \mathbf\Phi_{x}(z) \vcw,\\
    \vcu &= \mtK (zI - A - B\mtK)^{-1} \vcw =: \mathbf\Phi_u(z) \vcw,
    \end{aligned}
\end{equation}
where $\resp_x(z) \in \fdC^{N \times N}$ and $\resp_u(z) \in \fdC^{N \times N}$ are system responses that map the disturbance $\vcw$ to state $\vcx$ and control input $\vcu$, respectively. The following SLS theorem states that all achievable responses lie in an affine subspace of strictly proper stable rational transfer functions $\frac{1}{z} \mathcal R \mathcal H_\infty$.

\begin{theorem} \emph{\cite[Thm. 4.1]{anderson_system_2019}}
For the LTI system evolving under the dynamics \eqref{eq:linearDynamics} and control policy $\vcu = \mtK \vcx$, the following statements are true:
\begin{enumerate}
    \item The affine subspace defined by
    \begin{equation}\label{eq:achievability-constraint}
        \begin{bmatrix} zI-A & -B \end{bmatrix} \begin{bmatrix} \mathbf\Phi_x \\ \mathbf\Phi_u    \end{bmatrix} = I, \quad \mathbf\Phi_x, \mathbf\Phi_u \in \frac{1}{z}\mathcal R \mathcal H_\infty
    \end{equation}
    parameterizes all system responses from $\vcw$ to $(\vcx, \vcu)$ as defined in \eqref{eq:system-response}, achievable by an internally stabilizing state feedback controller $\mtK$.
    \item For any transfer matrices $\mathbf\Phi_x, \mathbf\Phi_u$ satisfiying \eqref{eq:achievability-constraint}, the controller $\mtK = \mathbf\Phi_u \mathbf\Phi_x^{-1}$ is internally stabilizing and achieves the desired system response in \eqref{eq:system-response}.
\end{enumerate}
\label{thm:SLS}    
\end{theorem}

For disturbance $w(t) \overset{\text{i.i.d.}}{\sim} \mathcal{N}(0,I)$, one can recast the optimization problem \eqref{eq:LQG} in terms of system responses $\mathbf\Phi_{x}$ and $\mathbf\Phi_u$ as
% eq:SLSOptimization
\begin{equation} \label{eq:SLSOptimization}
\begin{aligned}
    \underset{\mathbf\Phi_{x} , \mathbf\Phi_{u}}{\text{minimize}} &\ \quad 
    \htwonorm{\mtQ^{1/2}\mathbf\Phi_{x}}^{2} + \htwonorm{\mtR^{1/2}\mathbf\Phi_{u}}^{2} \\
    \text{s.t. } &\ \quad \text{constraint \eqref{eq:achievability-constraint}.}
\end{aligned}
\end{equation}
With a slight abuse of notation, we define $J(\resp_x, \resp_u)$ to be the LQR cost achieved by $\resp_x$ and $\resp_u$ in the objective of \eqref{eq:SLSOptimization} and $J(\mtK)$ as the cost \eqref{eq:quadraticCost} achieved by applying controller $\mtK$.

\subsection{SLS for Graph Symmetric Systems}
We now proceed to show that under Assumption \ref{assump:diagonalizableCost}, the optimal system response for a graph symmetric system that solves the LQR problem \eqref{eq:SLSOptimization} can be written as a graph filter.

\begin{theorem} \label{thm:optimalLinear}
Given a GMD $S = V \Lambda_{S} V^{\top}$, consider an instance of the LQR problem \eqref{eq:SLSOptimization} where the underlying system and cost satisfy Assumption~\ref{assump:diagonalizableCost}. Then, there exists a global optimum $(\optresp_x, \optresp_u)$ where both $\optresp_x$ and $\optresp_u$ are diagonalizable by $\mtV$, i.e.,
\begin{equation*}
\begin{aligned}
    \optresp_x &= V \mtLambda_x^* V^{\top}, \quad  \optresp_u &= V \mtLambda_u^* V^{\top},
\end{aligned}
\end{equation*}
where $\mtLambda_x^*$ and $\mtLambda_u^*$ are diagonal transfer matrices. Hence, the optimal controller $\mtK^\opt=(\optresp_u)(\optresp_x)^{-1}$ can also be diagonalized by $V$.
\end{theorem}
\begin{proof}
See Appendix.
\end{proof}

\begin{remark}
Note that the elements defined by the diagonal responses $\mtLambda^\opt_x, \mtLambda^\opt_u$ are \emph{transfer functions} $[\mtLambda^\opt_x]_{ii}(z), [\mtLambda^\opt_u]_{ii}(z)$. Thus the resulting graph filter taps are transfer functions as well, i.e., a transfer function $\resp(z)$ that is simultaneously diagonalizable with the matrix $S$ can be written as:
\begin{equation} \label{eq:optimalLinear}
    \resp(z) = \sum_{k=0}^{N-1}\boldsymbol\phi_k(z)S^k.
\end{equation}
\end{remark}

The main implication of Theorem~\ref{thm:optimalLinear} is that the optimal linear state-feedback controller for graph symmetric systems under Assumption \ref{assump:diagonalizableCost} is a graph filter and can thus be implemented via distributed message passing. We note, however, that the above result implies that the resulting optimal system response $\optresp := (\optresp_x, \optresp_u)$ could be dense, as $S^{N-1}$ is dense if $S$ defines a connected graph. This can be undesirable in practice, as it requires $N-1$ communication exchanges with one-hop neighbors, potentially causing significant delays if the size $N$ of the graph is large. In the next section, we leverage a robust variant of the SLS parameterization given in Theorem \ref{thm:SLS} to restrict the optimal system responses to only the first $F\ll N-1$ filter taps while guaranteeing stability and near optimal performance.

We end this section by noting that in the graph signal processing literature \citep{Sandryhaila2013-DSPG}, a controller of the form of \eqref{eq:optimalLinear} is known as a linear, shift-invariant (LSI) graph filter, and is analogous to an LTI filter.
Note that $S \optresp \vcx = \optresp S \vcx$, hence the name.
In particular, Equation \eqref{eq:optimalLinear} is a spatially finite impulse response (FIR) graph filter \citep{Segarra2017-GraphFilterDesign} that is completely characterized by a finite set of $N$ filter taps that can be conveniently described by a collection of transfer functions $\boldsymbol\phi^{\opt} = [\boldsymbol\phi_{0}^{\opt},\ldots,\boldsymbol\phi_{N-1}^{\opt}]^{\Tr} \in \fdR^{N}$.  We emphasize that the transfer functions themselves are not restricted to be temporally FIR.
Spatially FIR graph filters are also known as convolutional graph filters \citep{Ruiz2021-GNN} due to their sum-and-shift nature, understanding that the effect of the operation $S \vcx$ is to shift the signal around the graph (thus, oftentimes, the GMD $S$ is referred to as the graph shift operator).
Furthermore, spatially FIR graph filters satisfy the convolution theorem that indicates that a convolution in the vertex domain can be computed by means of an elementwise multiplication in the spectrum domain \citep{Sandryhaila2013-DSPG}.
Finally, in the context of finite graphs, it is observed that the space of FIR graph filters of the form \eqref{eq:optimalLinear}, characterized by $N$ filter taps, is equivalent to the space of spatially infinite impulse response (IIR) graph filters as well as autoregressive, moving average (ARMA) graph filters \citep{Isufi2017-ARMA}.

\section{Localized Approximations to the Optimal Distributed Linear Controller} \label{sec:approx}
In this section, we discuss several methods to approximate the optimal dense system response in the form of \eqref{eq:optimalLinear} with one that is localized and uses ${F \ll N}$ filter taps. We start with a projection method based on graph filter design. We then present a robust SLS formulation of the approximation problem that guarantees the stability of the resulting localized controller. Finally, we show how these two can be combined into a robust projection method that also ensures stability.  In the following, we define $\optresp := (\optresp_x, \optresp_u)$, and recall that $\optresp$ can be written as a graph filter \eqref{eq:optimalLinear} defined by transfer function filter taps $\boldsymbol\phi^\opt_k(z)$, $k=0,...,N-1$.  We further recall that each transfer function filter tap $\boldsymbol\phi^\opt_k(z)$ admits the following expansion in terms of its Markov parameters: $\boldsymbol\phi^\opt_k(z) = \sum_{i=1}^\infty z^{-i}\phi^\opt_k[i]$.

\subsection{Naive Projection}
We propose an approach inspired by the graph signal processing literature, wherein we exploit the graph filter structure of the optimal system responses \citep{Segarra2017-GraphFilterDesign}. More specifically, we project the optimal system responses $\optresp$ onto graph filters of order $F$ in the $\mathcal H_2$ norm by solving the following optimization problem
% eq:minGraphFilter
\begin{equation} \label{eq:minGraphFilter}
    \min_{\boldsymbol\phi} \htwonorm{ \resp - \optresp }^2.
\end{equation}
Here $\boldsymbol\phi := [\boldsymbol\phi_{0}(z),\ldots,\boldsymbol\phi_{F-1}(z)]^{\top} \in \fdC^{F}$ collects the $F$ transfer function filter taps defining $\resp := \sum_{k=0}^{F-1} \boldsymbol \phi_k(z)S^k$. If we further restrict each transfer function filter tap  $\phi_k(z)$ to be FIR of order $n$, i.e., if we write $\boldsymbol\phi_k(z) = \sum_{i=1}^n z^{-i}\phi_k[i]$, then this reduces to solving the following unconstrained quadratic program
\begin{equation}\label{eq:minGraphFilterFIR}
    \min_{\{\phi_k[i]\}} \sum_{i=1}^n \left\|\sum_{k=0}^{F-1}\phi_k[i]S^k - \sum_{l=0}^{N-1}\phi^\opt_l[i]S^l \right\|_F^2.
\end{equation}
 %%%%%%%%%%%%%%%%%%%%%%%%%%%%%%%%%%%%%%%%
%%%%          Proposition
%%%%%%%%%%%%%%%%%%%%%%%%%%%%%%%%%%%%%%%%

% prop:filterApproximation
\begin{proposition}[Approximating filter taps] \label{prop:filterApproximation}
If the eigenvalues $\{\lambda_{i}\}$ of the GMD $\mtS$ are all distinct, then the filter taps that solve \eqref{eq:minGraphFilter} are given by
\begin{equation}
    \phi_{k}[i] = \phi_{k}^{\opt}[i] + \vceps_k[i]
\end{equation}
where $\vceps$ is the error vector computed as
\begin{equation}
    \vceps[j] = \Big( \sum_{i=1}^{N} \vclambda_{iF} \vclambda_{iF}^{\Tr} \Big)^{-1} \Big( \sum_{i=1}^{N} \vclambda_{iF} \vclambda_{i(N-F)}^{\Tr} \Big) \boldsymbol\phi_{N-F}^{\opt}[j]
\end{equation}
with $\vclambda_{iF} := [1,\lambda_{i},\lambda_{i}^{2},\ldots,\lambda_{i}^{F-1}] \in \fdR^{F}$ is the collection of the first $F$ powers of $\lambda_{i}$, $\vclambda_{i(N-F)} := [\lambda_{i}^{F},\ldots,\lambda_{i}^{N-1}] \in \fdR^{N-F}$ is the collection of the remaining powers and $\boldsymbol\phi_{N-F}^{\opt}[i] := [\phi^\opt_{F}[i],\ldots,\phi^\opt_{N-1}[i]] \in \fdR^{N-F}$ collects the tail $N-F$ optimal filter taps.
\end{proposition}
%%%% PROOF
%
\begin{proof}
    It follows from using the convexity of \eqref{eq:minGraphFilterFIR}, matrix calculus and properties of Vandermonde matrices.
\end{proof}
%
%%%%       End of Proposition
%%%%%%%%%%%%%%%%%%%%%%%%%%%%%%%%%%%%%%%%

Prop.~\ref{prop:filterApproximation} determines in closed-form how to compute the filter of order $F$ that best approximates the optimal linear distributed controller in the $\mathcal{H}_2$ norm. It also shows that each filter tap transfer function $\boldsymbol\phi_{k}(z)$ is obtained as the optimal tap $\boldsymbol\phi_{k}^{\opt}(z)$ with an added corrective term that accounts for the $N-F$ taps that could not be included.

This approach is easy to implement computationally, as it only requires solving a least squares problem to minimize the projection cost. However, the resulting controller cannot be guaranteed to be stabilizing. As we show later via numerical simulation, approximations with a small number of filter taps $F$ are often unstable. This motivates an approach that takes into account the stability of the resulting controller.

\subsection{Localized Approximations via Robust SLS}
%One way to guarantee the stability of the truncated controller is to use robust-SLS-based approaches.
Robust SLS \citep{matni2017scalable, anderson_system_2019} offers a systematic way to reason about \emph{approximate system responses}, i.e., system responses that do not exactly satisfy the achievability constraint \eqref{eq:achievability-constraint}. In particular, as shown in the following result, robust SLS allows for an explicit characterization of the effects of using approximate system responses for controller design.

\begin{theorem}[Corollary 4.4 of \cite{anderson_system_2019}]\label{thm:robustSLS}
Let $(\resp_x, \resp_u, \mathbf{\Delta})$ be a solution to
\begin{equation}\label{eq:robust-achievability-constraint}
    \begin{bmatrix} zI-A & -B  \end{bmatrix} \begin{bmatrix} \resp_x \\ \resp_u    \end{bmatrix} = I + \mathbf\Delta, \quad \resp_x, \resp_u \in \frac{1}{z}\rhinfty.
\end{equation}
Then if $\lVert \mathbf\Delta \rVert_{\mathcal H_\infty} < 1$ the controller $\mathbf{K} = \resp_u\resp_x^{-1}$ stabilizes the system \eqref{eq:linearDynamics}, and the actual system response that is achieved is given by
\[
\begin{bmatrix} \vcx \\ \vcu \end{bmatrix} = \begin{bmatrix} \resp_x \\ \resp_u \end{bmatrix} (I + \mathbf\Delta)^{-1} \vcw.
\]
\end{theorem}

We leverage this result to provide an upper-bound on the amount of truncation that can be applied to $\optresp$ without destabilizing the system.

\begin{corollary}\label{cor:robust-cost-bound}
Let $(\resp_x, \resp_u, \mathbf{\Delta})$ be a solution to \eqref{eq:robust-achievability-constraint} and assume that $\hinfnorm{\mathbf\Delta} < 1$. Then the controller $\mtK = \resp_u\resp_x^{-1}$ achieves an LQR cost \eqref{eq:SLSOptimization} $J$ that can be bounded as
\begin{equation*}
    J(\mtK) \leq \frac{1}{1-\hinfnorm{\mathbf\Delta}} \htwonorm{\begin{bmatrix} Q^{1/2} & 0 \\ 0 & R^{1/2} \end{bmatrix} \begin{bmatrix} \resp_x \\ \resp_u \end{bmatrix}}.
\end{equation*}
\end{corollary}
\begin{proof}
First, we note that by Theorem \ref{thm:robustSLS}, the system responses $(\tilderesp_x, \tilderesp_u)$ achieved by $\mtK$ are given by
\[\begin{bmatrix} \tilderesp_x \\ \tilderesp_u \end{bmatrix} = \begin{bmatrix} \resp_x \\ \resp_u \end{bmatrix} (I + \mathbf\Delta)^{-1}.\]
Thus, the cost achieved by the controller $\mtK$ is bounded by
\begin{align*}
    J(\mtK) &= J(\tilderesp_x, \tilderesp_u)\\
    &= \htwonorm{\begin{bmatrix} Q^{1/2} & 0 \\ 0 & R^{1/2} \end{bmatrix} \begin{bmatrix} \resp_x \\ \resp_u \end{bmatrix} (I + \mathbf\Delta)^{-1}}\\
    &\leq \hinfnorm{(I + \mathbf\Delta)^{-1}} \htwonorm{\begin{bmatrix} Q^{1/2} & 0 \\ 0 & R^{1/2} \end{bmatrix} \begin{bmatrix} \resp_x \\ \resp_u \end{bmatrix}}\\
    &\leq \frac{1}{1-\hinfnorm{\mathbf\Delta}} \htwonorm{\begin{bmatrix} Q^{1/2} & 0 \\ 0 & R^{1/2} \end{bmatrix} \begin{bmatrix} \resp_x \\ \resp_u \end{bmatrix}},
\end{align*}
where we used Cauchy-Schwarz in the first inequality and the small gain theorem and the fact that $\hinfnorm{\mathbf\Delta} < 1$ in the last step.
\end{proof}
Corollary \ref{cor:robust-cost-bound} offers a way to synthesize \textit{stable} truncated system responses. %that might not be able to exactly satisfy the achievability constraints \eqref{eq:achievability-constraint}.
Specifically, to synthesize a system response that uses only $F < N-1$ filter taps while guaranteeing both stability and performance, we propose the following optimization problem:
% eq:robustSLSOptimization
\begin{subequations} \label{eq:robustSLSOptimization}
\begin{align}
\begin{split}
    \underset{\boldsymbol\phi_{x}, \boldsymbol\phi_{u}, \gamma\in(0,1)}{\text{minimize}} &\ \quad 
    \frac{1}{1-\gamma} \htwonorm{ \begin{bmatrix}  Q^{1/2} & 0 \\ 0 & R^{1/2}    \end{bmatrix}
    \begin{bmatrix}     \resp_x \\ \resp_u    \end{bmatrix} }
\end{split}\\
\begin{split}\label{eq:robustConstraint}
    \text{s.t. } & \quad
    \begin{bmatrix}    zI - A & -B   \end{bmatrix}
    \begin{bmatrix}     \resp_x \\ \resp_u    \end{bmatrix} = I + \mathbf\Delta, \\
    &\ \hinfnorm{\mathbf\Delta} \leq \gamma,\\
    &\ \resp_x(z) = \sum_{k=0}^{F-1} \boldsymbol\phi_{x,k}(z) S^k, \, \boldsymbol\phi_{x,k}(z) \in \frac{1}{z}\mathcal{RH}_\infty,\\
    &\ \resp_u(z) = \sum_{k=0}^{F-1} \boldsymbol\phi_{u,k}(z) S^k, \, \boldsymbol\phi_{u,k}(z) \in \frac{1}{z}\mathcal{RH}_\infty.
\end{split}
\end{align}
\end{subequations}

By Corollary \ref{cor:robust-cost-bound}, the solution to \eqref{eq:robustSLSOptimization} defines a controller that is stabilizing.  We further show in the next result that it enjoys guaranteed suboptimality bounds relative to the optimal controller defined by $\optresp$.  We first introduce the following notation: for a system response of the form $\resp = \sum_{k=0}^{N-1} \boldsymbol\phi_k(z)S^k$, we define the $F$-truncation $P_F(\resp)$ and the $F$-tail $P_{F\perp}(\resp)$ as
\begin{align*}
    P_F(\resp) := \sum_{k=0}^{F-1} \boldsymbol\phi_k(z)S^k,\
    P_{F\perp}(\resp) := \sum_{k=F}^{N-1} \boldsymbol\phi_k(z)S^k.
\end{align*}
\begin{theorem}
Let $(\tilderesp_x, \tilderesp_u, \tilde{\gamma})$ be the optimal solution to the robust SLS problem \eqref{eq:robustSLSOptimization}. Let $(\optresp_x, \optresp_u)$ be the optimal solution to the untruncated SLS problem \eqref{eq:SLSOptimization}. Suppose that
\[
\hinfnorm{\mathbf\Delta^*} := \hinfnorm{ (zI - A)P_{F\perp}(\optresp_x) - B P_{F\perp}(\optresp_u) } < 1.
\]
Then, the controller $\mtK = \tilderesp_u \tilderesp_x^{-1}$ is stabilizing and
\begin{equation}\label{eq:suboptimalityBound}
    J(\mtK) \leq \frac{1}{1 - \hinfnorm{\mathbf\Delta^*}} \Big( J(\optresp_x, \optresp_u) + J(P_{F\perp}(\optresp_x), P_{F\perp}(\optresp_u))    \Big)
\end{equation}
\end{theorem}
\begin{proof}
By the constraints \eqref{eq:robustConstraint} and Theorem \ref{thm:robustSLS}, we immediately have that $\mtK$ is stabilizing. To show the given suboptimality bound, we first note that there exist some $\gamma$ such that $(P_F(\optresp_x), P_F(\optresp_u), \gamma)$ is a feasible solution to the robust optimization problem \eqref{eq:robustSLSOptimization}. To see this, observe that
\begin{align*}
    & \begin{bmatrix} zI-A & -B \end{bmatrix} \begin{bmatrix} P_F(\optresp_x) \\ P_F(\optresp_u) \end{bmatrix} \\
    =\ &\begin{bmatrix} zI-A & -B \end{bmatrix} \begin{bmatrix} \optresp_x \\ \optresp_u \end{bmatrix} - \begin{bmatrix} zI-A & -B \end{bmatrix} \begin{bmatrix} P_{F\perp}(\optresp_x) \\ P_{F\perp}(\optresp_u) \end{bmatrix}\\
    =\ & I - \mathbf\Delta^*,
\end{align*}
where the last step follows from the achievability of the optimal response $(\optresp_x, \optresp_u)$ and the definition of $\mathbf\Delta^*$. Thus, $(P_F(\optresp_x), P_F(\optresp_u), \hinfnorm{\mathbf\Delta^*})$ is a feasible solution with our assumption that $\hinfnorm{\mathbf\Delta^*} < 1$. Denote the robust SLS objective \eqref{eq:robustSLSOptimization} as $\tilde{J}(\resp_x, \resp_u, \gamma)$. By the optimality of the solution $(\tilderesp_x, \tilderesp_u, \tilde{\gamma})$, we have that
\begin{align*}
 \tilde{J}(\tilderesp_x, &\tilderesp_u, \tilde{\gamma}) \leq \tilde{J}(\optresp_x, \optresp_u, \hinfnorm{\mathbf\Delta^*}) \\
 &\leq \frac{1}{1-\|\mathbf\Delta^*\|_{H_\infty}} \htwonorm{ \begin{bmatrix}  Q^{1/2} & 0 \\ 0 & R^{1/2}    \end{bmatrix}    \begin{bmatrix}  P_F(\optresp_x) \\ P_F(\optresp_u)\end{bmatrix}},
\end{align*}
where we applied Corollary \ref{cor:robust-cost-bound} in the second inequality. The desired result follows then from the fact that
\[
\begin{bmatrix}  P_F(\optresp_x) \\ P_F(\optresp_u)\end{bmatrix} = \begin{bmatrix}  \optresp_x \\ \optresp_u\end{bmatrix} - \begin{bmatrix}  P_{F\perp}(\optresp_x) \\ P_{F\perp}(\optresp_u)\end{bmatrix}
\]
and an application of the triangle inequality.
\end{proof}
This optimization problem is jointly quasi-convex and can be solved efficiently using bisection. Further, feasibility provides a stability certificate in the form of $\hinfnorm{\mathbf\Delta}<1$.

\subsection{Robust Projection}
Lastly, we can combine the robustness constraints used in robust SLS with the signal-processing-based projection method. Specifically, we solve the following optimization problem
\begin{equation} \label{eq:robustProjection}
\begin{aligned}
    \underset{\boldsymbol\Phi_x , \boldsymbol\Phi_{u}, \gamma\in(0,1)}{\text{minimize}} &\ \quad
    \htwonorm{\resp - \resp^*}^2 \\
    \text{s.t. } &\ \quad (\boldsymbol\Phi_x, \boldsymbol\Phi_u, \gamma)\;\text{satisfy constraint \eqref{eq:robustConstraint}}.
\end{aligned}
\end{equation}
We note that solving this problem does not give an upper bound on the cost of the resulting controller, but the robustness constraint ensures that the resulting controller is stabilizing.

\subsection{Implementation}
We end this section by detailing practical implementation details for optimization problems \eqref{eq:robustSLSOptimization} and \eqref{eq:robustProjection}. First, we note that computationally, one cannot directly optimize for the IIR system responses as is written in \eqref{eq:robustSLSOptimization}, \eqref{eq:robustProjection}. In practice, we use an FIR approximation of the strictly proper transfer functions $\boldsymbol\phi_x(z)$ and $\boldsymbol\phi_u(z)$, i.e., $\boldsymbol\phi(z)$ is parameterized as
\[
\boldsymbol\phi(z) = \sum_{i=1}^{n} z^{-i}\phi[i]
\]
for some given FIR order $n$. As shown in \cite{anderson_system_2019}, the suboptimality incurred by such an FIR approximation decays exponentially in the horizon $n$.

The $\mathcal H_\infty$-norm constraints on the (also FIR) transfer matrix $\mathbf\Delta$ can then be enforced via semidefinite programming (see Theorem 5.8 in \cite{dumitrescu2007positive}), potentially introducing a nontrivial computational burden. However, we note that one can replace the $\mathcal H_\infty$-norm constraint in optimization problems \eqref{eq:robustSLSOptimization} and \eqref{eq:robustProjection} with any induced norm constraint.  A particularly appealing option is the $\ell_1 \to \ell_1$ induced norm (which defines the $\mathcal{L}_1$-norm of the transpose system), as this norm decomposes columnwise.  As shown in \cite{anderson_system_2019, wang_separable_2018}, the resulting robustness constraints are linear and embarrassingly parallelizable. We defer this extension to future work.
\section{Numerical Experiments} \label{sec:sims}
We show that our approach offers a principled way to trade off performance and communication complexity through numerical experiments. We also demonstrate the importance of the robustness constraints in synthesizing stable distributed controllers and compare the performance of robust SLS and projection-based methods on synthesizing localized controllers. All code needed to reproduce the examples found in this section is available at \url{https://github.com/unstable-zeros/graph-symmetric-systems}.

\subsection{Setup}
In the following experiments, we consider the distributed linear quadratic regulator (LQR) problem \eqref{eq:LQG} over $N=10$ scalar subsystems. We generate the GMD $S$ and dynamic matrices $A$ and $B$ using a process similar to that in \cite{Gama2022-DistributedLQR}. To generate a problem instance, we start by creating the communication network $\mathcal{G}$ by randomly sampling $N$ numbers $\{u_i\}_{i=1}^N\sim{}U[0,1]$, and creating a bi-directional link between $v_i$ and each of its $3$ nearest points as defined by the topology on the interval $[0,1]$ under the metric $d(v_i,v_j)=|u_i-u_j|$. We then take $S$ to be a symmetric matrix that shares the sparsity pattern of the Laplacian of $\mathcal G$, with its entry values sampled independently from $\mathcal N(0,1)$. The GMD $S$ is then normalized to have a spectral radius of $1$. We generate the dynamics matrices $A$ and $B$ to share the same eigenvectors as $S$, and sample their eigenvalues i.i.d.\ from the standard normal distribution -- hence both $A$ and $B$ are symmetric matrices. We take the cost matrices $Q = R = I_{N}$. For both of the following experiments, we randomly generate $50$ problem instances using this process. We end by noting that $\mathcal{G}$ generated this way have, on average, a diameter of $5.92$ hops.

For the implementation of the optimization problems, we approximate the transfer functions $\boldsymbol\phi(z)$ with an FIR horizon of $n=10$. Further, for the robust SLS problem \eqref{eq:robustSLSOptimization}, instead of using bisection to determine the best value of $\gamma$, we fix $\gamma=0.98$, as empirically the value of $\gamma$ does not significantly affect the cost achieved by the controllers.

\begin{figure}[htbp]
\centerline{\includegraphics[width=0.5\textwidth]{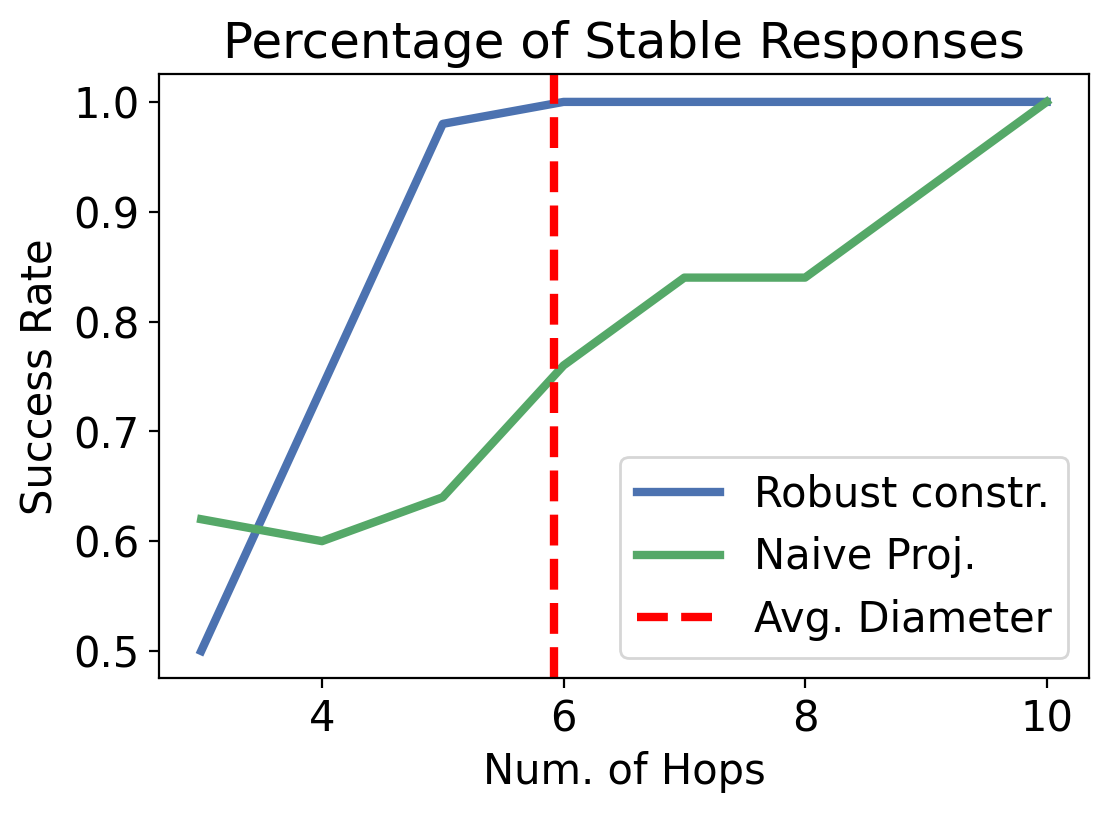}}
\caption{The percentage of synthesized $F$-hop controllers that are stabilizing using naive projection and robust SLS-based synthesis across 50 random trials.  The dashed red line denotes the average graph diameter of systems.}
\label{fig:trunc10-successrate}
\end{figure}

\subsection{Importance of Stability Constraints}
In this experiment, we demonstrate the importance of the robust SLS-based stability constraints in synthesizing stable distributed controllers. We vary the number of allowed filter taps and apply both the naive projection \eqref{eq:minGraphFilter} and robust SLS \eqref{eq:robustProjection} methods to $50$ randomly generated problem instances. For the naive projection method, we report the percentage of resulting controllers that are stable. For robust SLS, we report the percentage of optimization problems \eqref{eq:robustSLSOptimization} being feasible, as feasible solutions optimization problem \eqref{eq:robustSLSOptimization} are guaranteed to be stabilizing. The results are shown in Figure~\ref{fig:trunc10-successrate}.

First, we observe that as expected, a higher number of filter taps result in a higher probability of synthesizing stable responses for both methods. However, the naive projection method has nonzero probability of resulting in an unstable controllers even with a large number of filter taps i.e., even when the projection error between the \textit{stable} optimal system responses and the projected responses is small. %This fragility is highly undesirable in practice. 
On the other hand, the percentage of stable solutions resulting from the robust SLS problem increases monotonically with the number of filter taps, which is expected for a principled way of synthesizing stable controllers. Robust SLS also generally achieves a higher percentage of certifiably stable responses than that of naive projection, except for in the extremely sparse case of $F=3$. This suggests that the robust constraint might be too restrictive for synthesizing extremely sparse responses. Combined with the low computation cost of naive projection, this suggests a potential benefit of applying both methods in the sparse regime.

\begin{figure}[htbp]
\centerline{\includegraphics[width=0.5\textwidth]{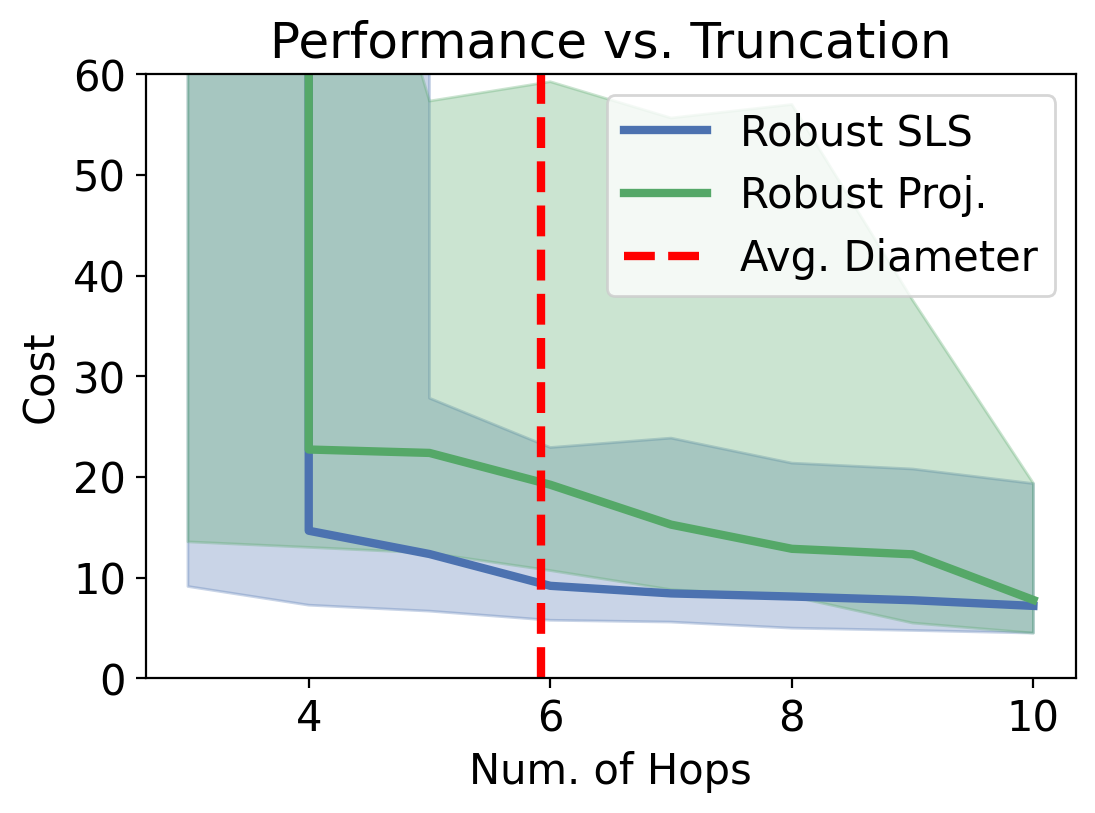}}
\caption{Cost achieved by $F$-hop controllers across 50 randomly generated systems. Solid lines denote the median cost achieved by each method, and the shaded regions show the 25-th and 75-th percentile of the costs.}
\label{fig:trunc10-perf}
\end{figure}

\subsection{Truncation Performance}
In this experiment, we compare the performance of robust SLS and robust projection for different filter tap numbers $F$ on the same $50$ randomly generated problem instances. We show the median (solid lines), $25$-th and $75$-th percentile (shaded regions) of the costs achieved by both methods in Figure~\ref{fig:trunc10-perf}.

First, we note that the median costs decreases monotonically for both methods as the number of hops increase. This shows that that the optimization problems can leverage the increase in expressivity of the graph filters to achieve better performance, which matches our intuition. 
%This is, on the other hand, also slightly surprising as neither method has any guarantees on such a monotonic decrease in achieved costs (robust SLS provides a decrease guarantee for the upper bound surrogate objective). 
Second, we note that the robust SLS-based method achieves a lower cost than robust Projection over for all numbers of filter taps considered. We also note that to the left of $4$ hops, the upper boundary of the shaded region, which represents the $75$-th percentile of the cost, is infinite, indicating that at least $25\%$ of the robust synthesis problems are infeasible. This again suggests a need to develop more flexible methods in the sparse regime.
\section{Conclusion}
In this works, we introduced the notion of graph symmetric systems and showed that for linear quadratic problems, the optimal system response for graph symmetric systems can be written as (potentially dense) graph filters. We then proposed three methods to approximate the optimal responses with localized responses and validated their performance in numerical simulation. Directions of future work include relaxing the GSS constraints, applying the results on $\mathcal L_1$ norm to enable distributed computation, and understanding how this can better inform GNN-based controllers with nonlinear activation functions.

%%%%%%%%%%%%%%%%%%%%%%%%%%%%%%%%%%%%%%%%
%%%%                                %%%%
%%%%         ACKNOWLEDGMENT         %%%%
%%%%                                %%%%
%%%%%%%%%%%%%%%%%%%%%%%%%%%%%%%%%%%%%%%%

%\section*{Acknowledgment}

%%%%%%%%%%%%%%%%%%%%%%%%%%%%%%%%%%%%%%%%
%%%%                                %%%%
%%%%           REFERENCES           %%%%
%%%%                                %%%%
%%%%%%%%%%%%%%%%%%%%%%%%%%%%%%%%%%%%%%%%

% \vfill\pagebreak

% \balance

\bibliographystyle{unsrtnat}
\bibliography{bibFiles/myIEEEabrv,bibFiles/biblioDistributedLinear}

\appendix
\subsection{Proof for Theorem \ref{thm:optimalLinear}}
We show that for any optimal system response $\optresp_x$, $\optresp_u$ of the optimization problem \eqref{eq:SLSOptimization} that is not diagonalizable with $V$, i.e.,
\begin{equation*}
    \begin{aligned}
    \mathbf\Lambda^\opt_x &= V^\top \optresp_x V, \quad
    \mathbf\Lambda^\opt_u  = V^\top \optresp_u V
    \end{aligned}
\end{equation*}
are not diagonal, we can construct a simultaneously diagonalizable system response $\resp_x'$, $\resp_u'$ that is equally optimal.  In particular, we construct such a system response as follows:
% eq:diag_phi
%
\begin{gather}
    \resp_x' = V \mathbf\Lambda_x' V^{\top}, \quad
     \resp_u' = V \mathbf\Lambda_u' V^{\top}\\
    [\mathbf\Lambda_x']_{ij} =
    \begin{cases}
        [\mathbf\Lambda^\opt_x]_{ij} & i=j \\
        0, & \text{o.w.}
    \end{cases},\;
    [\mathbf\Lambda_u']_{ij} =
    \begin{cases}
        [\mathbf\Lambda^\opt_u]_{ij} & i=j \\
        0 & \text{o.w.}
    \end{cases}.\label{eq:diag_phi}
\end{gather}

We first show that $\resp_x'$, $\resp_u'$ are feasible solutions of \eqref{eq:SLSOptimization}. From the achievability condition on $\optresp_x, \optresp_u$, we have that
\begin{equation}
\begin{aligned}
    (zI - A) \optresp_x - B \optresp_u &= I.
\end{aligned}
\end{equation}
Using the simultaneous diagonalizability of $A$ and $B$, we have
\begin{equation}
    V (zI - \Lambda_{A}) V^{\top} V \mtLambda_x^\opt V^{\top} -
    V \Lambda_{B} V^{\top} V \mtLambda_u^\opt V^{\top} = I
    \implies (zI - \Lambda_{A}) \mtLambda^\opt_x - \Lambda_{B} \mtLambda^\opt_u = I.
\end{equation}
Since the matrices $(zI - \Lambda_{A})$, $\Lambda_{B}$ and $I$ are diagonal, we have
\begin{equation}
    (zI - \Lambda_{A}) \mtLambda_x' - \Lambda_{B} \mtLambda_u'
    = \diag\Big[(zI - \Lambda_{A}) \mtLambda^\opt_x - \Lambda_{B} \mtLambda^\opt_u\Big]
    = I,
\end{equation}
where $\diag[\mtM]$ projects a matrix $\mtM$ onto its diagonal elements. Therefore, $\resp_x'$, $\resp_u'$ is also feasible.

Now, we show that $\resp_x'$, $\resp_u'$ gives a cost at least as good as that of $\optresp_x, \optresp_u$. By the simultaneous diagonalizability of the matrices $Q$ and $R$, and the fact that the $\mathcal H_2$-norm is invariant under unitary transformations, we have that
\begin{equation}\label{eq:spectral-cost}
\begin{aligned}
    &\htwonorm{Q^{1/2}\resp_x'}^2 + \htwonorm{R^{1/2}\resp_u'}^2 \\
    =\;& \htwonorm{V \Lambda_Q^{1/2} V^\top V\mtLambda_x' V^\top} + \htwonorm{V \Lambda_R^{1/2} V^\top V\mtLambda_u' V^\top}\\
    =\;& \htwonorm{\Lambda_{Q}^{1/2} \mtLambda_x'}^{2} + \htwonorm{\Lambda_{R}^{1/2} \mtLambda_u'}^{2}.\\
\end{aligned}
\end{equation}
Denoting the $i$-th eigenvalue of $\mtQ$ and $\mtR$ with $\lambda_{Q,i}, \lambda_{R,i}$, respectively, we have the inequality
\begin{equation}
\begin{aligned}
    &\htwonorm{\Lambda_{Q}^{1/2} \mtLambda_x'}^{2} + \htwonorm{\Lambda_{R}^{1/2} \mtLambda_u'}^{2}.\\
    =\;&\sum_{i,j} \lambda_{Q, i} \htwonorm{[\mtLambda_x']_{ji}}^2 + \sum_{i,j} \lambda_{R, i} \htwonorm{[\mtLambda_u']_{ji}}^2\\
    \leq\;& \sum_{i,j} \lambda_{Q, i} \htwonorm{[\mtLambda_x^\opt]_{ji}}^2 + \sum_{i,j} \lambda_{R, i} \htwonorm{[\mtLambda_u^\opt]_{ji}}^2\\
    =\;&\htwonorm{\Lambda_{Q}^{1/2} \mtLambda_x^*}^{2} + \htwonorm{\Lambda_{R}^{1/2} \mtLambda_u^*}^{2},\\
\end{aligned}
\end{equation}
which follows from the definition of $\mtLambda_x'$ and $\mtLambda_u'$ in equation \eqref{eq:diag_phi}. Reversing the steps in \eqref{eq:spectral-cost}, we see that $\resp_x', \resp_u'$ achieves a cost at least as good as that of $\optresp_x, \optresp_u$. We can thus conclude that there always exists an optimal simultaneously diagonalizable system response to the LQR problem \eqref{eq:SLSOptimization}. The controller $\mtK' = (\resp_u')(\resp_x')^{-1}$ is thus optimal and simultaneously diagonalizable by $V$.

\subsection{GSS and Controllers Satisfy Quadratic Invariance}
Here we show that optimal control problems over graph symmetric systems and controllers satisfy quadratic invariance \citep{rotkowitz2005characterization}.  Before proceeding, we remark that the analysis of LQR optimal control problem over GSSs does not require quadratic invariance.  In particular, in Theorem \ref{thm:optimalLinear} we analyze the \emph{unconstrained optimal control problem} and show that the resulting \emph{unconstrained optimal controller} satisfies a corresponding notion of graph symmetry.  However, in the interest of completeness, we show that if such a constraint were imposed on the controller during synthesis, the resulting problem satisfies quadratic invariance.

To that end, the corresponding constrained controller synthesis problem can be stated as
\begin{equation*}
    \begin{aligned}
        \underset{\mtK}{\text{minimize}}&\quad J(\mtK)\\
        \text{subject to}&\quad \mtK\ \text{stabilizes system \eqref{eq:linearDynamics}}\\
        &\quad \mtK \in \mathcal{S}
    \end{aligned}
\end{equation*}
where $\mathcal{S} := \{S \in \mathcal{RH}_{\infty}\ |\ S\ \text{is diagonalizable by } V\}.$ Denoting the plant input-output transfer function as 
$$\mtG(z) = (zI-A)^{-1}B,$$
we have the following proposition.
\begin{proposition}
The set of graph symmetric controllers $\mathcal{S}$ is quadratically invariant under $\mtG$ if system \eqref{eq:linearDynamics} is graph symmetric.
\end{proposition}
\begin{proof}
The proof follows directly from the definition of quadratic invariance \cite[Def. 2]{rotkowitz2005characterization} by straightforward calculation. First, we note that
$$\mtG(z) = V(zI-\mtLambda_A)^{-1}\mtLambda_BV^\top$$
is diagonalizable by $V$. For any controller $\mtK \in \mathcal{S}$, it then follows immediately that
$$\mtK \mtG \mtK \in \mathcal{S}$$
as the product of simultaneously diagonalizable matrices is also simultaneously diagonalizable, proving the claim.
\end{proof}

\end{document}